\documentclass[12pt]{article}
\usepackage[utf8]{inputenc}
\usepackage[russian,english]{babel}
\usepackage{graphicx}
\usepackage{amssymb,amsmath,amsthm,color,url}
\newtheorem*{theorem}{Theorem}
\newtheorem*{proposition}{Proposition}
\newtheorem*{lemma}{Lemma}
\newtheorem*{remark}{Remark}
\newtheorem*{corollary}{Corollary}

\def\eps{\varepsilon}

\title{Compressibility and probabilistic proofs}
\author{Alexander Shen\thanks{%
LIRMM CNRS \& University of Montpellier. On leave from IITP RAS, Moscow, Russia. E-mail address: \hbox{\texttt{alexander.shen@lirmm.fr}}. Supported by ANR-15-CE40-0016-01 RaCAF grant.}}

\begin{document}
\maketitle

\begin{abstract}
We consider several examples of probabilistic existence proofs using  compressibility arguments, including some results that involve Lov\'asz local lemma.
\end{abstract}

\section{Probabilistic proofs: a toy example}

There are many well known probabilistic proofs that objects with some properties exist. Such a proof estimates the probability for a random object to violate the requirements and shows that it is small (or at least strictly less than $1$). Let us look at a toy example.

Consider a $n\times n$ Boolean matrix and its $k\times k$ minor (the intersection of $k$ rows and $k$ columns chosen arbitrarily). We say that the minor is \emph{monochromatic} if all its elements are equal (either all zeros or all ones). 

\begin{proposition}
For large enough $n$ and for $k=O(\log n)$, there exists a $(n\times n)$-matrix that does not contain a monochromatic $(k\times k)$-minor.
\end{proposition}

\begin{proof}
We repeat the same simple proof three times, in three different languages.

(Probabilistic language) Let us choose matrix elements using independent tosses of a fair coin. For a given $k$ colums and $k$ rows, the probability of getting a monochromatic minor at their intersection is $2^{-k^2+1}$. (Both zero-minor and one-minor have probability $2^{-k^2}$.) There are at most $n^k$ choices for columns and the same number for rows, so by the union bound the probability of getting at least one monochromatic minor is bounded by 
$$
n^k \times n^k \times 2^{-k^2+1}= 2^{2k\log n - k^2 + 1}=2^{k(2\log n-k)+1}
$$
and the last expression is less then $1$ if, say, $k=3\log n$ and $n$ is suffuciently large.

(Combinatorial language) Let us count the number of bad matrices. For a given choice of columns and rows we have $2$ possibilities for the minor and $2^{n^2-k^2}$ possibilities for the rest, and there is at most $n^k$ choices for raws and columns, so the total number of matrices with monochromatic minor is
$$
n^k \times n^k \times 2\times 2^{n^2-k^2}=2^{n^2+2k\log n-k^2+1}=2^{n^2+k(2\log n - k)+1},
$$
and this is less than $2^{n^2}$, the total number of Boolean $(n\times n)$-matrices.

(Compression language) To specify the matrix that has a monochromatic minor, it is enough to specify $2k$ numbers between $1$ and $n$ (rows and column numbers), the color of the monochromatic minor ($0$ or $1$) and the remaining $n^2-k^2$ bits in the matrix (their positions are already known). So we save $k^2$ bits (compared to the straightforward list of all $n^2$ bits) using $2k\log n+1$ bits instead (each number in the range $1\ldots n$ requires $\log n$ bits; to be exact, we may use $\lceil \log n\rceil$), so we can compress the matrix with a monochromatic minor if $2k\log n+1 \ll k^2$, and not all matrices are compressible. 
\end{proof}

Of course, these three arguments are the same: in the second one we multiply probabilities by $2^{n^2}$, and in the third one we take logarithms.
However, the compression language provides some new viewpoint that may help our intuition.

\section{A bit more interesting example}

In this example we want to put bits (zeros and ones) around the circle in a ``essentially asymmetric'' way: each rotation of the circle should change at least a fixed percentage of bits. More precisely, we are interested in the following statement:

\begin{proposition}
There exists $\eps>0$ such for every suffuciently large $n$ there exists a sequence $x_0 x_1\ldots x_{n-1}$ of bits such that for every $k=1,2,\ldots,n-1$ the cyclic shift by $k$ positions produces a sequence
$$
y_0=x_k, y_1=x_{k+1},\ldots,y_{n-1}=x_{k-1},
$$
that differs from $x$ in at least $\eps n$ positions \textup(the Hamming distance between $x$ and $y$ is at least $\eps n$\textup).
\end{proposition}

\begin{figure}[h]
\begin{center}
\includegraphics[scale=1.0]{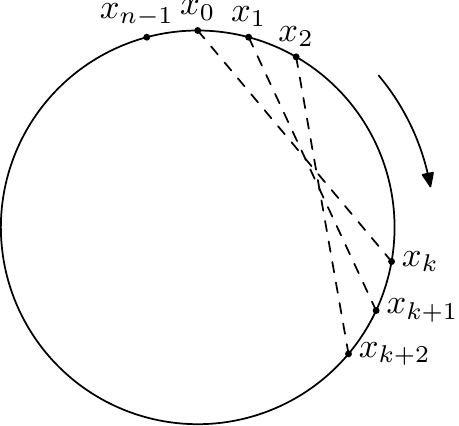}
\end{center}
\caption{A string $x_0\ldots x_{n-1}$ is bad if most of the dotted lines connect equal bits}\label{pic:rotation}
\end{figure}

\begin{proof}
Assume that some rotation (cyclic shift by $k$ positions) transforms $x$ into a string $y$ that coincides almost everywhere with $x$. We may assume that $k\le n/2$: the cyclic shift by $k$ positions changes as many bits as the cyclic shift by $n-k$ (the inverse one). Imagine that we dictate the string $x$ from left to right.  First $k$ bits we dictate normally. But then the bits start to repeat (mostly) the previous ones ($k$ positions before), so we can just say ``the same'' or ``not the same'', and if $\eps$ is small, we know that most of the time we say ``the same''. Technically, we have $\eps n$ different bits, and at least $n-k\ge n/2$ bits to dictate after the first $k$, so the fraction of ``not the same'' signals is at most $2\eps$. It is well known that strings of symbols where some symbols appear more often than others can be encoded efficiently. Shannon tells us that a string with two symbols with frequencies $p$ and $q$ (so $p+q=1$) can be encoded using 
$$
H(p,q)=p \log\frac{1}{p} +q \log \frac{1}{q} 
$$
bits per symbol and that $H(p,q)=1$ only when $p=q=1/2$. In our case, for small $\eps$, one of the frequencies is close to $0$ (at most $2\eps$), and the other one is close to $1$, so $H(p,q)$ is significantly less than $1$. So we get a significant compression for every string that is bad for the theorem, therefore most string are good (so good string do exist).

More precisely, every string $x_0\ldots x_{n-1}$ that does not satisfy the requirements, can be described by 
\begin{itemize}
\item $k$\qquad [$\log n$ bits]
\item $x_0,\ldots,x_{k-1}$\qquad [$k$ bits]
\item $x_k\oplus x_0, x_{k+1}\oplus x_1, \ldots, x_{n-1}\oplus x_{n-k-1}$ \qquad [$n-k$ bits where the fraction of $1$s is at most $2\eps$, compressed to $(n-k)H(2\eps,1-2\eps)$ bits]
\end{itemize}
For $\eps<1/4$ and for large enough $n$ the economy in the third part (compared to $n-k$) is more important than $\log n$ in the first part.
\end{proof}

Of course, this is essentially a counting argument: the number of strings of length $(n-k)$ where the fraction of $1$s is at most $2\eps$, is bounded by $2^{H(2\eps,1-2\eps)(n-k)}$ and we show that the bound for the number of bad strings,
$$
\sum_{k=1}^{n/2} 2^{k} 2^{H(2\eps,1-2\eps)(n-k)}
$$ 
is less than the total number of strings ($2^n$). Still the compression metaphor makes the proof more intuitive, at least for some readers.

\section{Lov\'asz local lemma and\\ Moser--Tardos algorithm}

In our examples of probabilistic proofs we proved the existence of objects that have some property by showing that \emph{most} objects have this property (in other words, that the probability of this property to be true is close to $1$ under some natural distribibution). Not all probabilistic proofs go like that. One of the exceptions is the famous Lov\'asz local lemma~(see, e.g.,~\cite{alon-spencer}). It can be used in the situations where the union bound does not work: we have too many bad events, and the sum of their probabilities exceeds $1$ even if probability of each one is very small. Still Lov\'asz local lemma shows that these bad events do not cover the probability space entirely, assuming that the bad events are ``mainly independent''. The probability of avoiding these bad events is exponentially small, still Lov\'asz local lemma provides a positive lower bound for it. 

This means, in particular, that we cannot hope to construct an object satisfying the requirements by random trials, so the bound provided by Lov\'asz local lemma does not give us a randomized algorithm that constructs the object with required properties with probability close to $1$. Much later Moser and Tardos~\cite{moser,moser-tardos} suggested such an algorithm --- in fact a very simple one. In other terms, they suggested a different distribution under which good objects form a majority. 

We do not discuss the statement of Lov\'asz local lemma and Moser--Tardos algorithm in general. Instead, we provide two examples when they can be used, and the compression-language proofs that can be considered as ad hoc versions of Moser--Tardos argument. These two examples are (1)~satisfiability of formulas in conjunctive normal form (CNF) and (2)~strings without forbidden factors.

\section{Satisfiable CNF}

A CNF (\emph{conjunctive normal form}) is a propositional formula that is a conjuction of \emph{clauses}. Each clause is a disjunction of \emph{literals}; a literal is a propositional variable or its negation. For example, CNF
$$
(\lnot p_1 \lor p_2 \lor p_4)\land (\lnot p_2 \lor p_3\lor \lnot p_4)
$$
consists of two clauses. First one prohibits the case when
$p_1 = \textsc{true}$, $p_2=\textsc{false}$, $p_4=\textsc{false}$; the second one prohibits the case when $p_2=\textsc{true}$, $p_3=\textsc{false}$, $p_4=\textsc{true}$. A CNF is \emph{satisfiable} if it has a \emph{satisfying assigment} (that makes all clauses true, avoiding the prohibited combinations). In our example there are many satisfying assigments. For example, if $p_1=\textsc{false}$ and $p_3=\textsc{true}$, all values of other variables are OK.

We will consider CNF where all clauses include $n$ literals with $n$ different variables (from some pool of variables that may contain much more than $n$ variables). For a random assignment (each variable is obtained by an independent tossing of a fair coin) the probability to violate a clause of this type is $2^{-n}$ (one of $2^n$ combinations of values for $n$ variables is forbidden). Therefore, \emph{if the number of clauses of this type is less than $2^n$, then the formula is satisfiable}. This is a tight bound: using $2^n$ clauses with the same variables, we can forbid all the combinations and get an unsatisfiable CNF.

The following result says that we can guarantee the satisfiability for formuli with much more clauses. In fact, the total number of clauses may be arbitrary (but still we consider finite formulas, of course). The only thing we need is the ``limited dependence'' of clauses. Let us say that two clauses are \emph{neighbors} if they have a common variable (or several common variables). The clauses that are not neighbors correspond to independent events (for a random assignment). The following statement says that if the number of neighbors of each clause is bounded, then CNF is guaranteed to be satisfisable.

\begin{proposition}
Assume that each clause in some CNF contains $n$ literals with different variables and has at most $2^{n-3}$ neighbor clauses. Then the CNF is satisfiable.
\end{proposition}

Note that $2^{n-3}$ is a rather tight bound: to forbid all the combinations for some $n$ variables, we need only $2^n$ clauses.

\begin{proof}
It is convenient to present a proof using the compression language, as suggested by Lance Fortnow. Consider the following procedure $\textsc{Fix}(C)$ whose argument is a clause (from our CNF).
\begin{flushleft}
\qquad \{ $C$ is false \}\\
\qquad $\textsc{Fix}(C)$:\\
\qquad\qquad $\textsc{Resample}(C)$\\
\qquad\qquad \textbf{for} all $C'$ that are neighbors of $C$:\\
\qquad\qquad\qquad \textbf{if} $C'$ is false \textbf{then} $\textsc{Fix}(C')$\\
\qquad\{ $C$ is true; other clauses that were true remain true \}
\end{flushleft}
Here $\textsc{Resample}(C)$ is the procedure that assigns fresh random values to all variables in $C$. The pre-condition (the first line) says that the procedure is called only in the situation where $C$ is false. The post-condition (the last line) says that \emph{if the procedure terminates}, then $C$ is true after termination, and, moreover, all other clauses of our CNF that were true before the call remain true. (The ones that were false may be true or false.)

Note that up to now we do not say anything about the termination: note that the procedure is randomized and it may happen that it does not terminate (for example, if all \textsc{Resample} calls are unlucky to choose the same old bad values).

\textbf{Simple observation}: if we have such a procedure, we may apply it to all clauses one by one and after all calls (assuming they terminate and the procedure works according to the specification) we get a satisfying assignment.

\textbf{Another simple observation}: it is easy to prove the ``conditional correctness'' of the procedure $\textsc{Fix}(C)$. In other words, it achieves its goal assuming that (1)~it terminates; (2)~all the recursive calls $\textsc{Fix}(C')$ achieve their goals. It is almost obvious: the $\textsc{Resample}(C)$ call may destroy (=make false) only clauses that are neighbors to $C$, and all these clauses are \textsc{Fix}-ed after that. Note that $C$ is its own neighbor, so the \textbf{for}-loop includes also a recursive call $\textsc{Fix}(C)$, so after all these calls (that terminate and satisfy the post-condition by assumption) the clause $C$ and all its neighbors are true and no other clause is damaged. 

Note that the last argument remains valid even if we delete the only line that really changes something, i.e., the line $\textsc{Resample}(C)$. In this case the procedure never changes anything but still is conditionally correct; it just does not terminate if one of the clauses is false.
 
It remains to prove that the call $\textsc{Fix}(C)$ terminates with high probability. In fact, it terminates with probability $1$  if there are no time limits and with probability exponentially close to $1$ in polynomial time. To prove this, one may use a compression argument: we show that \emph{if the procedure works for a long time without terminating, then the sequence of random bits used for resampling is compressible}. We assume that each call of $\textsc{Resample}()$ uses $n$ fresh bits from the sequence. Finally, we note that this compressibility may happen only with exponentially small probability. 

Imagine that $\textsc{Fix}(C)$ is called and during its recursive execution performs many calls
$$
\textsc{Resample}(C_1),\ldots,\textsc{Resample}(C_N)
$$
(in this order) but does not terminate (yet). We stop it at some moment and examine the values of all the variables.

\begin{lemma}
Knowing the values of the variables after these calls and the sequence $C_1,\ldots,C_N$, we can reconstruct all the $Nn$ random bits used for resampling.
\end{lemma}
 
\begin{proof}[Proof of the lemma] Let us go backwards. By assumption we know the values of all variables after the calls. The procedure $\textsc{Resample}(C_N)$ is called only when $C_N$ is false, and there is only one $n$-tuple of values that makes $C_N$ false. Therefore we know the values of all variables before the last call, and also know the random bits used for the last resampling (since we know the values of variables after resampling).

The same argument shows that we can reconstruct the values of variables before the preceding call $\textsc{Resample}(C_{N-1})$, and random bits used for the resampling in this call, etc.
\end{proof}

Now we need to show that the sequence of clauses $C_1,\ldots,C_N$ used for resampling can be described by less bits than $nN$ (the number of random bits used). Here we use the assumption saying each clause has at most $2^{n-3}$ neighbors and that the clauses $C'$ for which $\textsc{Fix}(C')$ is called from $\textsc{Fix}(C)$, are neighbors of $C$.

One could try to say that since $C_{i+1}$ is a neighbor of $C_i$, we need only $n-3$ bits to specify it (there are at most $2^{n-3}$ neighbors by assumption), so we save $3$ bits per clause (compared to $n$ random bits used by resampling). But this argument is wrong: $C_{i+1}$ is not always the neighbor of $C_i$, since we may return from a recursive call that causes resampling of $C_i$ and then make a new recursive call that resamples $C_{i+1}$. 

To get a correct argument, we should look more closely at the tree of recursive calls generated by one call $\textsc{Fix}(C)$ (Fig.~\ref{pic:treecall}). In this tree the sons of each vertex correspond to neighbor clauses of the father-clause.
\begin{figure}[h]
\begin{center}
\includegraphics[scale=1.0]{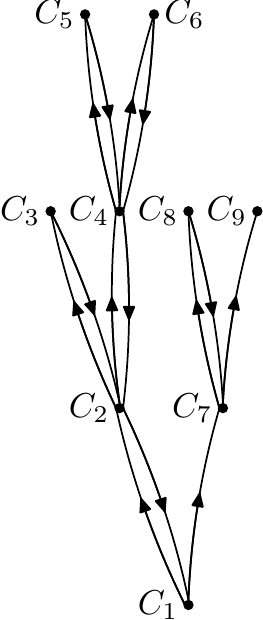}
\end{center}
\caption{The tree of recursive calls for $\textsc{Fix}(C_1)$ (up to some moment)}\label{pic:treecall}
\end{figure}
The sequence of calls is determined by a walk in this tree, but we go up and down, not only up (as we assumed in the wrong argument). How many bits we need to encode this walk (and therefore the sequence of calls)? We use one bit to distinguish between steps up and down. If we are going down, no other information is needed. If we are going up (and resample a new clause), we need one bit to say that we are going up, and $n-3$ bits for the number of neighbor we are going to. For accounting purposes we combine these bits with a bit needed to encode the step back (this may happen later or not happen at all), and we see that in total we need at most $(n-3)+1+1=n-1$ bits per each resampling. This is still less than $n$, so we save one bit for each resampling. If $N$ is much bigger than the number of variables, we indeed compress the sequence of random bits used for resampling, and this happens with exponentially small probability.

This argument finishes the proof.
\end{proof}

\section{Tetris and forbidden factors}

The next example is taken from word combinatorics. Assume that a list of binary strings $F_1,\ldots,F_k$ is given. These $F_i$ are considered as ``forbidden factors'': this means that we want to construct a (long) string $X$ that does not have any of $F_i$ as a factor (i.e., none of $F_i$ is a substring of $X$). This may be possible or not depending on the list. For example, if we consider two strings $0, 11$ as forbidden factors, every string of length $2$ or more has a forbidden factor (we cannot use zeros at all, and two ones are forbidden). 

The more forbidden factors we have, the more chances that they block the growth in the sense that every sufficiently long string has a forbidden factor. Of course, not only the number of factors matters: e.g., if we consider $0, 00$ as forbidden factors, then we have long strings of ones without forbidden factors. However, now we are interested in quantitative results of the following type: \emph{if the number of forbidden factors of length $j$ is $a_j$, and the numbers $a_j$ are ``not too big'', then there exists an arbitrarily long string without forbidden factors}.

This question can be analyzed with many different tools, including Lov\'asz local lemma (see~\cite{rumyantsev-ushakov}) and Kolmogorov complexity. Using a complexity argument, Levin proved that if $a_j=2^{\alpha j}$ for some constant $\alpha<1$, then there exists a constant $M$ and an infinite sequence that does not contain forbidden factors of length smaller than $M$.  (See~\cite[Section 8.5]{usv} for Levin's argument and other related results.) A nice sufficient condition was suggested by Miller~\cite{miller}: we formulate the statement for the arbitrary alphabet size.

\begin{proposition}
Consider an alphabet with $m$ letters. Assume that for each $j\ge 2$ we have $a_j$ ``forbidden'' strings of length $j$. Assume that there exist some constant $x>0$ such that
$$
\sum_{j\ge 2} a_j x^j < mx -1
$$
Then there exist arbitrarily long strings that do not contain forbidden substrings.
\end{proposition}

\textbf{Remarks}. 1. We do not consider $j=1$, since this means that some letters are deleted from the alphabet. 

2. By compactness the statement implies that there exists an infinite sequence with no forbidden factors.

3. The constant $x$ should be at least $1/m$, otherwise the right hand side is negative. This means that $a_j/m^j$ should be small, and this corresponds to our intution ($a_j$ should be significantly less than $m^j$, the total number of strings of length $j$).

The original proof from~\cite{miller} uses some ingenious potential function defined on strings: Miller shows that if its value is less than $1$, then one can add some letter preserving this property.  It turned out (rather misteriously) that exactly the same condition can be obtained by a completely different argument (following~\cite{goncalves,ochem}) ---
so probably the inequality is more fundamental than it may seem! This argument is based on compression. 

\begin{proof}
Here is the idea. We start with an empty string and add randomly chosen letters to its right end. If some forbidden string appears as a suffix, it is immediately deleted. So forbidden strings may appear only as suffixes, and only for a short time. After this ``backtracking'' we continue adding new letters. (This resembles the famous ``tetris game'' when blocks fall down and then disappear under some conditions.)

We want to show that if this process is unsuccessful in the sense that after many steps we still have a short string, then the sequence of added random letters is compressible, so this cannot happen always, and therefore a long string without forbidden factors exists. Let us consider a ``record'' (log file) for this process that is a sequence of symbols ``$+$'' and ``$+\langle\text{deleted string}\rangle$'' (for each forbidden string we have a symbol, plus one more symbol without a string). If a letter was added and no forbidden string appears, we just add `$+$' to the record. If we have to delete some forbidden string $s$ after a letter was added, we write this string in brackets after the $+$ sign. Note that we do \emph{not} record the added letters, only the deleted substrings. (It may happen that several forbidden suffixes appear; in this case we may choose any of them.)

\begin{lemma}
At every stage of the process the current string and the record uniquely determine the sequence of random letters used.
\end{lemma}

\begin{proof}[Proof of the lemma] Having this information, we can reconstruct the configuration going backwards. This reversed process has steps where a forbidden string is added (and we know which one, since it is written in brackets in the record), and also steps when a letter is deleted (and we know which letter is deleted, i.e., which random letter was added when moving forwards).
\end{proof}

If after many (say, $T$) steps we still have a short current string, then the sequence of random letters can be described by the record (due to the Lemma; we ignore the current string part since it is short). As we will see, the record can be encoded with less bits than it should have been (i.e., less than $T\log m$ bits). Let us describe this encoding  and show that it is efficient (assuming the inequality $\sum a_j x^j< mx-1$). 

We use arithmetic encoding for the lengths. Arithmetic encoding for $M$ symbols starts by choosing positive reals $q_1,\ldots,q_M$ such that $q_1+\ldots+q_M=1$. Then we split the interval $[0,1]$ into parts of length $q_1,\ldots,q_M$ that correspond to these $M$ symbols. Adding a new symbol corresponds to splitting the current interval in the same proportion and choosing the right subinterval. For example, the sequence $(a,b)$ corresponds to $b$th subinterval of $a$th interval; this interval has length $q_aq_b$. The sequence $(a,b,\ldots, c)$ corresponds to interval of length $q_aq_b\ldots q_c$ and can be reconstructed given any point of this interval (assuming $q_1,\ldots,q_M$ are fixed); to specify some binary fraction in this interval we need at most $-\log (q_aq_b\ldots q_c)+O(1)$ bits, i.e., $-\log q_a-\log q_b-\ldots-\log q_c+O(1)$ bits.

Now let us apply this technique to our situation. For $+$ without brackets we use $\log (1/p_0)$ bits, and for $+\langle s\rangle$  where $s$ is of length $j$, we use $\log (1/p_j)+\log a_j$ bits. Here $p_j$ are some positive reals to be chosen later; we need $p_0+\sum p_j =1$. Indeed, we may split $p_j$ into $a_j$ equal parts (of size $p_j/a_j$) and use these parts as $q_s$ in the description of arithmetical coding above;  splitting adds $\log a_j$ to the code length for strings of length $j$.

To bound the total number of bits used for encoding the record, we perform amortised accounting and show that the average number of bits per letter is less than $\log m$. Note that the number of letters is equal to the number of $+$ signs in the record. Each $+$ without brackets increases the length of the string by one letter, and we want to use less that $\log m - c$ bits for its encoding, where $c>0$ is some constant saying how much is saved as a reserve for amortized analysis.  And $+\langle s\rangle$ for a string $s$ of length $j$ decreases the length by $j-1$, so we want to use less than $\log m+c(j-1)$ bits (using the reserve).

So we need:
\begin{align*}
\log (1/p_0) &< \log m - c;\\
\log (1/p_j)+\log a_j &< \log m +c(j-1)
\end{align*}
together with 
$$
p_0+\sum_{j\ge 2} p_j =1.
$$

Technically is it easier to use non-strict inequalities in the first two cases and a strict one in the last case (and then increase $p_i$ a bit):
$$
\log (1/p_0) \le \log m - c;\ \log (1/p_j)+\log a_j \le \log m +c(j-1);\ p_0+\sum_{j\ge 2} p_j <1.
$$
Then for a given $c$ we take minimal possible $p_i$:
\begin{align*}
p_0 &= \frac{1}{m 2^{-c}}\\
p_j& = \frac{a_j (2^{-c})^j}{m2^{-c}}
\end{align*}
and it remains to show that the sum is less than $1$ for a suitable choice of $c$. Let $x=2^{-c}$, then the inequality can be rewritten as
$$
  \frac{1}{mx}+\sum_{j\ge 2} \frac{a_jx^j}{mx}< 1,
$$
or
$$
 \sum_{j\ge 2} a_jx^j < mx-1,
$$
and this is our assumption.

Now we see the role of this mystical $x$ in the condition: it is just a parameter that determines the constant used for the amortised analysis.
\end{proof}

\textbf{Acknowledgement}.
Author thanks his LIRMM colleagues, in particular Pascal Ochem and Daniel Gon\c calves, as well as the participants of Kolmogorov seminar in Moscow.

\section*{Appendix}

There is one more sufficient condition for the existence of arbitrarily long sequences that avoid forbidden substrings. Here is it.\footnote{A more general algebraic fact about ideals in a free algebra with $m$ generators is sometimes called Golod theorem; N.~Rampersad in \url{https://arxiv.org/pdf/0907.4667.pdf} gives a reference to Rowen's book (L.~Rowen, Ring Theory, vol.~II, Pure and Applied Mathematics, Academic Press, Boston, 1988, Lemma~6.2.7). This more general statement concerns ideals generated not necessarily by strings (products of generators), but by arbitrary uniform elements. The original paper is: 
\begin{otherlanguage*}{russian}
Е.С.\,Голод, И.Р.\,Шафаревич, О башне полей классов, Известия АН СССР, серия математическая, 1964, 
\end{otherlanguage*}
\textbf{28}:2, 261--272, 
\url{http://www.mathnet.ru/links/f17df1a72a73e5e73887c19b7d47e277/im2955.pdf}.} 
If the power series for
$$
\frac{1}{1-mx+a_2x^2+a_3x^3+\ldots}
$$
(where $a_i$ is the number of forbidden strings of length $m$) has all positive coefficients, then there exist arbitrarily long strings withour forbidden substrings. Moreover, in this case the number of $n$-letter strings without forbidden substrings is at least $g_n$, where $g_n$ is the $n$th coefficient of this inverse series. 

To prove this result, consider the number $s_k$ of allowed strings of length $k$.  It is easy to see that
$$
s_{k+1}\ge s_km - s_{k-1}a_2 - s_{k-2}a_3-\ldots - s_1a_k -s_0a_{k+1}.
$$   
Indeed, we can add each of $m$ letters to each of $s_k$ strings of length $k$, and then we should exclude the cases where there is a forbidden string at the end. This forbidden string may have length $2$, then there are at most $s_{k-1}a_2$ possibilities, or length $3$, there are at most $s_{k-2}a_3$ possibilities, etc. (Note that $s_0=1$ and $s_1=m$; note also that we can get a string with two forbidden suffixes, but this is OK, since we have an inequality.) These inequalities can be rephrased as follows: the product
$$
(1+mx+s_2x^2+s_3x^3+\ldots)(1-mx+a_2x^2+a_3x^3+\ldots)
$$
has only non-negative coefficients. Denote the second term by $A$; if 
$$1/A=1+mx+g_2x^2+g_3x^3+\ldots$$ 
has only positive coefficients $g_i$, (as our assumption says), then the first term is a product of two series with non-negative coefficients. The first factor ($1/A$) starts with $1$, so the $n$th coefficient of a product, i.e., $s_n$, is not less than $n$th coefficient of the second factor, i.e., $g_n$.

Surprisingly, this condition is closely related to the one considered above, as shown by Dmitry Piontkovsky (his name has a typo in the publication: 
\begin{otherlanguage*}{russian}
Д.И.\,Пиотковский, О росте градуированных алгебр с небольшим числом определяющих соотношений, УМН,%
\end{otherlanguage*}
1993,\textbf{48}:3(291), 199--200, \url{http://www.mathnet.ru/links/6034910939adb12fff0cd8fb9745dfc8/rm1307.pdf}): 

\begin{proposition}
Series
$$
\frac{1}{1-mx+a_2x^2+a_3x^3+\ldots}
$$
has all positive coefficients if and only if the series in the denominator has a root on a positive part of real line.
\end{proposition}

\begin{proof}
Assume that the series in the denominator does not have a root, but the inverse series has all positive coefficients. In fact, non-negative coefficients are enough to get a contradiction. For a series with all non-negative coefficients, or with finitely many negative coefficients, the radius of convergence is determined by behavior of the sum on the real line: when the argument approaches the convergence radius, the sum of the series goes to infinity.  Now we have the product of two series 
$$
 (1-mx+a_2x^2+a_3x^3+\ldots)(1+mx+g_2x^2+g_3x^3+\ldots)=1
$$
that is equal to $1$. One of these series should have finite convergence radius, otherwise both are everywhere defined and the product is everywhere $1$, but both are large for large $x$. Look at the minimal convergence radius (of two); one of the series goes to infinity near the corresponding point on the real line, so the other one converges to zero, so it has bigger convergence radius and reaches zero at the real line. Finally, note that only the first factor (the denominator) may have a zero, since the other one has all non-negative coefficients.

Now assume that the denominator has a zero; we have to prove that the inverse series has only positive coefficients. In general, the following result is true (D.~Piontkovsky): \emph{if the series 
$$
A=a_0+a_1x+a_2x^2+\ldots
$$
has $a_0>0$, and $a_2,a_3,\ldots\ge 0$, and for some positive $x$ this series converges to $0$, then the inverse series has all positive coefficients.} To prove this statement, let $\alpha$ be the root, so $A(\alpha)=0$. Recall the long division process that computes the inverse series. It produce the sequence of remainders: the first $R_0$ is $1$; then we subtract from the $k$th remainder
$$
R^{(k)}=R^{(k)}_k x^k+ R^{(k)}_{k+1} x^{k+1}+\ldots
$$
the product $(R^{(k)}/a_0)Ax^k$ to cancel the first term, and get the next remainder $R^{(k+1)}$. By induction we prove that for each remainder $R^{(k)}$:
\begin{itemize}
\item $R^{(k)}(\alpha)=1$;
\item all the coefficients $R^{(k)}_{k+1}$, $R^{(k)}_{k+2},\ldots$, except the first one, are negative or zeros;
\item the first coefficient $R^{(k)}_k$ is positive.
\end{itemize}

The first claim is true, since it was true for $R^{(k-1)}$ by induction assumption and we subtract a series that equals zero at $\alpha$. 

The second claim: by induction assumption the first coefficient in $R^{(k-1)}$ was positive, so we subtract the series $(R^{(k-1)}/a_0)Ax^{k-1}$ with positive first and non-negative third, fourth, etc. coefficients. The first term cancels the first term in $R^{(k-1)}$, the second term does not matter now, but all the subsequent coefficients are negative or zeros, since we subtract non-negative coefficients from non-positive ones.

Finally, the third claim is the consequence of the first two: if the sum is positive (equal to $1$) and all the terms except one are non-positive, then the remaining term is positive.

Therefore, all coefficients in the inverse series are positive.
\end{proof}

Note that we have shown that if all coefficients of the series $1/A$ are non-negative, then they are positive. Also note that we get a bit stronger result compared to the entropy argument where we required the series to reach a negative value (now the zero value is enough).

\end{document}